\documentclass[pra,longbibliography,twocolumn,showpacs,nofootinbib,superscriptaddress,notitlepage]{revtex4-2}
\usepackage{amsmath}
\usepackage{amssymb,bm}
\usepackage{amsthm}

\newtheorem{innercustomgeneric}{\customgenericname}
\providecommand{\customgenericname}{}
\newcommand{\newcustomtheorem}[2]{%
  \newenvironment{#1}[1]
  {%
   \renewcommand\customgenericname{#2}%
   \renewcommand\theinnercustomgeneric{##1}%
   \innercustomgeneric
  }
  {\endinnercustomgeneric}
}

\newcustomtheorem{customthm}{Theorem}
\newcustomtheorem{definitionBob}{Definition}

\usepackage{color,dsfont} 
\usepackage{graphicx}
\usepackage{subcaption}
\usepackage{ragged2e}
\DeclareCaptionJustification{justified}{\justifying}
\usepackage[colorlinks=true, hyperindex, breaklinks, linkcolor=blue, urlcolor=blue, citecolor=blue]{hyperref} 
\usepackage[normalem]{ulem}
\usepackage[capitalise]{cleveref}
\usepackage{mathrsfs}

\usepackage{mathtools}
\usepackage{float}
\usepackage{verbatim}
\usepackage{latexsym}
\usepackage{amsmath}
\usepackage{amssymb}
\usepackage{setspace}
\usepackage{amsfonts}
\usepackage{stmaryrd}
\usepackage{xcolor}
\usepackage{enumitem}
\usepackage{hhline}
\usepackage[normalem]{ulem}

\newtheorem{claim}{Claim} \newtheorem{definition}{Definition}

\newtheorem{fact}{Fact}

\newcommand{\codepar}[1]{\ensuremath{[\![#1]\!]}}

\crefname{equation}{Eq.\!}{Eqs.\!}
\crefname{figure}{Fig.\!}{Figs.\!}
\usepackage{multirow}
\mathchardef\mhyphen="2D

\begin{document}

\title{Fault-tolerant quantum error correction using error weight parities}

\author{Theerapat Tansuwannont}
\email{ttansuwannont@uwaterloo.ca}
\affiliation{
    Institute for Quantum Computing and Department of Physics and Astronomy,
    University of Waterloo,
    Waterloo, Ontario, N2L 3G1, Canada
    }
    
\author{Debbie Leung}
\email{wcleung@uwaterloo.ca}
\affiliation{
   Institute for Quantum Computing and Department of Combinatorics and Optimization,
    University of Waterloo,
    Waterloo, Ontario, N2L 3G1, Canada
    }
\affiliation{
	Perimeter Institute for Theoretical Physics, 
	Waterloo, Ontario, N2L 2Y5, Canada
}

\begin{abstract}
In quantum error correction using imperfect primitives, errors of high weight arising from a few faults are major concerns since they might not be correctable by the quantum error correcting code. Fortunately, some errors of different weights are logically equivalent and the same correction procedure is applicable to all equivalent errors, thus correcting high-weight errors is sometimes possible. In this work, we introduce a technique called weight parity error correction (WPEC) which can correct Pauli error of any weight in some stabilizer codes provided that the parity of the weight of the error is known. We show that the technique is applicable to concatenated codes constructed from the \codepar{7,1,3} Steane code or the \codepar{23,1,7} Golay code. We also provide a fault-tolerant error correction protocol using WPEC for the \codepar{49,1,9} concatenated Steane code which can correct up to 3 faults and requires only 2 ancillas.
\end{abstract}


\maketitle

\section{Introduction}
\label{sec:Intro}%

One crucial component for large-scale quantum computers is fault-tolerant error correction (FTEC), which suppresses error propagation throughout the circuits. An arbitrarily small logical error rate can be achieved through code concatenation, given that the physical error rate is below some constant threshold value \cite{Shor96,AB08,Preskill98,KLZ96,AGP06}. However, increasing overheads are needed for decreasing logical error rate \cite{Steane03,PR12,CJL16b,TYC17}. Conventional FTEC schemes require many ancillas during error syndrome measurements. For example, the Shor-style \cite{Shor96,DA07} and the Knill-style \cite{Knill05a} error corrections, which apply to any stabilizer code, require as many ancillas as the maximum weight of the stabilizer generators and twice the blocklength, respectively.  Steane-style error correction \cite{Steane97,Steane02} which applies to any CSS code requires one code block of ancillas.

Recently, several FTEC schemes that use only a few ancillas and are applicable to the \codepar{7,1,3} Steane code \cite{Steane96b} have been proposed. The scheme due to Yoder and Kim for the \codepar{7,1,3} code uses 2 ancillas (9 qubits in total) \cite{YK17}. Their ideas are further developed to a ``flag FTEC'' scheme which, for the \codepar{7,1,3} code, also uses 2 ancillas \cite{CR17a}. A flag FTEC scheme for any \codepar{n,k,d} stabilizer code requires $d+1$ ancillas \cite{CR20}, where the schemes for some specific code families may require fewer \cite{CR17a,CB18,TCL20,CKYZ20,CZYHZ20}. 
The flag technique that uses a few ancillas to detect high-weight errors can also be applied to various fault-tolerant schemes \cite{CR17b,CC19,SCC19,BCC19,BXG19,Vui18,GMB19,LA19,CN20,DB20,RBMS21}. Another FTEC scheme applicable to the \codepar{7,1,3} code was proposed by Reichardt; the scheme extracts several syndrome bits at once and requires no ancillas, provided that there are at least two code blocks (so at least 14 qubits are required in total) \cite{Reichardt18}.


In order to achieve an arbitrarily low error rate through code concatenation, the FTEC scheme used with the code must be modified accordingly.  One way to do this is replacing all physical qubits with code blocks and replacing all physical gates with corresponding logical gates \cite{AGP06}. For the \codepar{7,1,3} code, each qubit (including each ancilla qubit) required in an FTEC scheme will become a block of 7 physical qubits in the modified scheme.  Following this modification, the schemes in \cite{YK17,CR17a} applied to the \codepar{49,1,9} concatenated Steane code will require 63 qubits in total.  Meanwhile, the scheme in \cite{Reichardt18} requires 98 qubits in total, encoding 2 logical qubits.  Note that the maximum weight for the stabilizer generators increases quickly with concatenation.  These difficulties motivate our main question: how to reduce the number of ancillas required for an FTEC scheme for a concatenated code?

In this paper, we introduce a technique called weight parity error correction (WPEC) and construct an FTEC scheme for the \codepar{49,1,9} concatenated Steane code using only \emph{two} ancilla qubits. The scheme relies on the fact that, for the \codepar{7,1,3} code, errors with the same syndrome and weight parity differ by the multiplication of some stabilizer; these errors are thus logically equivalent and need not be distinguished from one another.  Therefore, error correction on each subblock of 7 qubits in the \codepar{49,1,9} code can be accomplished using only two ingredients: the error syndrome and the weight parity of error in each subblock.  Most importantly, the weight parity for each subblock of 7 qubits in the \codepar{49,1,9} code can be obtained from the full syndrome measurement.  Using this idea in conjunction with 2 ancilla qubits, our FTEC protocol for the \codepar{49,1,9} code can correct up to 3 faults. As a result, our protocol can suppress the error rate from $p$ to $O(p^4)$ using 51 qubits in total.

The paper is organized as follows: In \cref{sec:WPEC}, we observe the aforementioned equivalence between errors of any weight with the same syndrome and weight parity, and describe WPEC. In \cref{sec:Protocol}, we provide sufficient conditions for WPEC, then we provide syndrome extraction circuits and an FTEC protocol for the \codepar{49,1,9} concatenated Steane code using only two ancilla qubits. In \cref{sec:WPEC_Golay}, WPEC is extended to the \codepar{23,1,7} Golay code and concatenated Steane codes with more than 2 levels of concatenation. Last, we discuss our results and directions for future works in \cref{sec:Discussion}.

\section{Weight parity error correction for the Steane code}
\label{sec:WPEC}%
The Steane code \cite{Steane96b}, also known as the \codepar{7,1,3} code, is a quantum error correcting code that encodes 1 logical qubit into 7 physical qubits and can correct any error on up to 1 qubit. It has several desirable properties for fault-tolerant quantum computation, e.g., logical Clifford operations are transversal \cite{Shor96}. The Steane code is a code in the Calderbank-Shor-Steane (CSS) code family \cite{CS96,Steane96b} where $X$-type and $Z$-type errors can be detected and corrected separately.  The Steane code in the stabilizer formalism can be constructed from the parity check matrix of the classical $[7,4,3]$ Hamming code through the CSS construction \cite{Gottesman97}. In addition, it is known that any classical Hamming code can be rearranged into a cyclic code, a binary linear code in which any cyclic shift of a codeword is also a codeword \cite{MS77}. We can describe the Steane code in cyclic form with the following stabilizer generators:
\begingroup
\setlength\arraycolsep{1pt}	
\begin{equation}
\begin{matrix}
g^x_1: &X &I &X &X &X &I &I, & \quad & g^z_1: &Z &I &Z &Z &Z &I &I,\\
g^x_2: &I &X &I &X &X &X &I, & \quad & g^z_2: &I &Z &I &Z &Z &Z &I,\\
g^x_3: &I &I &X &I &X &X &X, & \quad & g^z_3: &I &I &Z &I &Z &Z &Z.
\end{matrix}
\end{equation}%
\endgroup
The generators of a stabilizer code define not only the codespace, but also the measurements that give rise to the error syndrome. When these measurements are imperfect, different sets of generators for the same code can have different fault-tolerant properties.
The use of the Steane code in cyclic form gives some advantages in distinguishing high-weight errors in consecutive form \cite{TCL20} (see \cref{sec:Discussion} for more details).
We can choose the logical $X$ and logical $Z$ operators to be $X^{\otimes 7} S$ and $Z^{\otimes 7} T$ for any stabilizer operators $S,T$.  
With this convention, we state the following crucial property of the Steane code that goes into our construction: 

\begin{fact}
Let $M$ be any $Z$-type operator (a tensor product of $I$s and $Z$s) defined on 7 qubits. Suppose $M$ commutes with all X-type generators of the \codepar{7,1,3} code. If $M$ has even weight, then it is a logical $I$; otherwise, if $M$ has odd weight, then it is a logical $Z$.
\label{fact:fact1}
\end{fact}
\begin{proof}
	Because $M$ is a $Z$-type operator that commutes with all $X$-type generators, $M$ is either a stabilizer of $Z$ type or a logical $Z$ operator.  
	Let $E_1$ and $E_2$ be $Z$-type operators with weights $w_1 $ and $w_2$. Then $E_1E_2$ is an operator of weight $w_1+w_2-2c$, where $c$ is the number of qubits supported by both $E_1$ and $E_2$. Observe that all stabilizer generators of the Steane code have even weight, and a multiplication of two operators with even weight always gives an operator with even weight. Thus, all stabilizers of $Z$ type (which are logical $I$ operators) have even weight. Moreover, a $Z$-type operator which is a logical $Z$ operator is of the form $Z^{\otimes 7} T$ where $T$ is a stabilizer of $Z$ type. Therefore, all logical $Z$ operators of $Z$ type have odd weight.
\end{proof}

For a Pauli error $E$ on a block of 7 qubits, the syndrome is a 6-bit string denoted by $s(E) = (s_x|s_z)$ where $s_x,s_z \in \mathbb{Z}_2^3$.  The $i$-th bit of $s_x$ (or $s_z$) is $0$ if $E$ commutes with $g_i^x$ (or $g_i^z$), and $1$ if $E$ anticommutes with $g_i^x$ ($g_i^z$).  If $E$ occurs to a codeword of the Steane code, $s(E)$ corresponds to the outcomes of measuring the six generators (0 and 1 correspond to $+1$ and $-1$ outcomes, respectively).
The Steane code is a \emph{perfect CSS code of distance $3$} 
meaning that for each $s_x$, $(s_x|000)$ is the syndrome of a \emph{unique}
weight-1 $Z$-type error, which we denote as $E^z_{wt\mhyphen 1}(s_x)$, 
and similarly each $(000|s_z)$ is the syndrome of a unique $X$-type
error \footnote{This is from the fact that the \codepar{7,1,3} Steane code can be constructed from the classical $[7,4,3]$ Hamming code which is a \emph{classical perfect code}, a code which saturates the classical Hamming bound \cite{MS77}.}.
For CSS codes, the $X$-type and $Z$-type error corrections are
independent of one another.  Furthermore, we focus on CSS codes in
which $X$-type and $Z$-type generators have the same form, and the same
method applies to both types of error correction.  So we focus on $Z$
errors for simplicity.  Since $Z$-type errors have trivial $s_z$, 
we focus on $s_x$ from now on.

With the above notations, consider the following simple error
correction procedure on the Steane code:
if the syndrome is $(s_x|000)$, do nothing if $s_x$ is trivial,
apply $E^z_{wt\mhyphen 1}(s_x)$ otherwise. We observe that if the
syndrome is caused by a $Z$-type error, then the procedure outputs
the encoded data transformed by a logical $I$ or logical $Z$.
This is because the actual $Z$-type error combined with the correction
remains $Z$-type and commutes with all of $g^x_{1,2,3}$, so the
conclusion follows from \cref{fact:fact1}. 
%

If a codeword is corrupted by an arbitrary $Z$-type error $E$, the
above procedure always recovers the codeword, but sometimes with an
undesirable logical $Z$ error.
The technique of weight parity error correction, to be developed 
next, is a revised procedure that will \emph{always} correct the error 
$E$, but it requires knowing whether $E$ has odd or even weight.
Measuring the error weight parity should not be done on a single layer of Steane code since it measures a logical operator on the Steane code. Fortunately, the parity information can be safely learnt for the constituent blocks when we concatenate the Steane code with itself.
We will describe these ideas in detail in the rest of this section,
and apply them for fault-tolerant error correction in the next
section.


First, we use \cref{fact:fact1} to show that $Z$-type errors with the same syndrome \emph{and} the same weight parity (whether odd or even) differ by the multiplication of some stabilizer.

\begin{claim}{Logical equivalence of errors with the same syndrome and weight parity for the \codepar{7,1,3} code}
	
Suppose $E_1,E_2$ are arbitrary $Z$-type errors (of any weights) on the \codepar{7,1,3} code with the same syndrome. Then, $E_1$ and $E_2$ have the same weight parity iff $E_1 = E_2S$ for some stabilizer $S$.
\label{claim:equiv_Steane}%
\end{claim}

\begin{proof}
Let $w_1,w_2$ be the weights of $E_1, E_2$, respectively. Let $N = E_1 E_2$ (so $E_2 = E_1 N$ as $E_1 = E_1^\dagger$). The weight of $N$ is equal to $w_1+w_2-2c$ where $c$ is the number of qubits supported by both $E_1$ and $E_2$. As $N$ commutes with all of $g^x_{1,2,3}$, from \cref{fact:fact1}, $N$ is a logical $I$ if and only if $w_1+w_2-2c$ is even (when $E_1$ and $E_2$ have the same weight parity).
\end{proof}



Second, we use \cref{claim:equiv_Steane} to provide a method for error
correction of $Z$-type error of arbitrary weight on the Steane code,
\emph{if the weight parity of the error is known}:

\begin{definition}{Weight parity error correction (WPEC) for the \codepar{7,1,3} code}
	
Suppose a $Z$-type error $E$ occurs to a codeword of the \codepar{7,1,3} code. Let $s_x$ and $w$ be the syndrome and the weight of $E$, $E^z_{wt\mhyphen 1}(s_x)$ be the weight-1 $Z$-type operator with syndrome $s_x$, and $E^z_{wt\mhyphen 2}(s_x)$ be any weight-2 $Z$-type operator with syndrome $s_x$, respectively. The following procedure is called \emph{weight parity error correction (WPEC):
\begin{enumerate}
	\item If $s_{x}$ is trivial, do nothing if $w$ is even, or apply any logical $Z$ if $w$ is odd.
	\item If $s_{x}$ is nontrivial, apply $E^z_{wt\mhyphen 1}(s_{x})$ if $w$ is odd, or apply $E^z_{wt\mhyphen 2}(s_{x})$ if $w$ is even.
\end{enumerate}}
\label{Def:WPEC}%
\end{definition}
WPEC always returns the original codewords because in each case, the error $E$ and the correction operation have the same syndrome and weight parity, so by \cref{claim:equiv_Steane}, the correction is logically equivalent to $E$.  

WPEC allows us to correct high-weight errors in the Steane code, but we need to know the weight parity of the error. The weight parity of a $Z$-type error is the outcome of measuring $X^{\otimes 7}$, so learning the weight parity is equivalent to a logical $X$ measurement, which can destroy the superposition of the logical state.
Fortunately, there is a setting in which the weight parity can be obtained without affecting the encoded data.  
Consider code concatenation in which each qubit of an error correcting code ${\cal C}_2$ is encoded into another quantum error correcting code ${\cal C}_1$.  If ${\cal C}_1$ is chosen to be the Steane code, the weight parity of each codeblock can potentially be learnt from the syndrome of ${\cal C}_2$.  We will develop WPEC for the concatenated Steane code in the rest of this section and show the advantage in the context of fault tolerance in the next section.

Consider code concatenation using two Steane codes in cyclic form. The resulting code which is a \codepar{49,1,9} code can be described by 48 stabilizer generators.
The first group of 42 generators, called 1st-level generators, have the form $g^x_i\otimes I^{\otimes 42},g^z_i\otimes I^{\otimes42},I^{\otimes 7} \otimes g^x_i\otimes I^{\otimes 35},I^{\otimes 7} \otimes g^z_i\otimes I^{\otimes 35},\dots,I^{\otimes 42}\otimes g^x_i,I^{\otimes 42}\otimes g^z_i$ for $i=1,2,3$.
The remaining 6 of these generators, called 2nd-level generators, have the form
\begingroup
\setlength\arraycolsep{1pt}	
\begin{equation}
\begin{matrix}
\tilde{g}^x_1: &\mathbf{X} &\mathbf{I} &\mathbf{X} &\mathbf{X} &\mathbf{X} &\mathbf{I} &\mathbf{I}, & \quad & \tilde{g}^z_1: &\mathbf{Z} &\mathbf{I} &\mathbf{Z} &\mathbf{Z} &\mathbf{Z} &\mathbf{I} &\mathbf{I},\\
\tilde{g}^x_2: &\mathbf{I} &\mathbf{X} &\mathbf{I} &\mathbf{X} &\mathbf{X} &\mathbf{X} &\mathbf{I}, & \quad & \tilde{g}^z_2: &\mathbf{I} &\mathbf{Z} &\mathbf{I} &\mathbf{Z} &\mathbf{Z} &\mathbf{Z} &\mathbf{I},\\
\tilde{g}^x_3: &\mathbf{I} &\mathbf{I} &\mathbf{X} &\mathbf{I} &\mathbf{X} &\mathbf{X} &\mathbf{X}, & \quad & \tilde{g}^z_3: &\mathbf{I} &\mathbf{I} &\mathbf{Z} &\mathbf{I} &\mathbf{Z} &\mathbf{Z} &\mathbf{Z},
\end{matrix}
\end{equation}%
\endgroup 
where $\mathbf{I}=I^{\otimes 7},\mathbf{X}=X^{\otimes 7},$ and $\mathbf{Z}=Z^{\otimes 7}$.  The logical $X$ and logical $Z$ operators
can be chosen to be $\bar{X}=X^{\otimes 49} S$ and $\bar{Z}=Z^{\otimes 49} T$ for any stabilizer operators $S, T$.
Relevant parts of the error syndrome corresponding to the 1st-level
and the 2nd-level generators will be called 1st-level and 2nd-level
syndromes, respectively.

Let us consider error correction on the \codepar{49,1,9} code and assume for now that error syndromes are reliable (which can be obtained from repetitive measurements). First, consider a simple motivating example, and suppose that a $Z$-type error $E$ acts nontrivially on at most one subblock of 7-qubit code. In order to perform WPEC, the weight parity of $E$ and the subblock in which $E$ occurs must be known. Suppose that $E$ has nontrivial 1st-level syndrome. The subblock in which $E$ occurs is actually the subblock whose corresponding 1st-level syndrome is nontrivial, while the weight parity of $E$ is a measurement result from a 2nd-level generator which acts nontrivially on that subblock (note that the 2nd-level generator must \emph{act nontrivially on all qubits} in such a subblock, thus a choice of 2nd-level generators is important). Now, suppose that $E$ has trivial 1st-level syndrome. The subblock in which $E$ occurs can no longer be identified by the 1st-level syndrome. Fortunately, since the 2nd-level Steane code (${\cal C}_2$) is a distance-3 code, it can identify if any of the $7$ subblocks of \codepar{7,1,3} code (the ${\cal C}_1$ subblocks) has a $Z$-type error logically equivalent to $Z^{\otimes 7}$, thus providing the weight parity for each subblock of \codepar{7,1,3} code. That is, if $E$ has trivial 1st-level syndrome and its weight is odd, the weight parity of $E$ and the subblock in which $E$ occurs can be determined using only the 2nd-level syndrome. (If $E$ has trivial 1st-level syndrome and has even weight, it is a stabilizer and no error correction is needed.)

%
%

If the $Z$-type error is more general and may act on multiple
subblocks, the 2nd-level syndrome may not provide the weight parities
of the subblocks.
Instead, we consider only $Z$-type errors that arise from a small
number of faults in specially designed generator measurements.  We
will show that for these errors, the weight parity for each subblock can
be determined by the 2nd-level syndrome along with the information
whether each subblock has trivial 1st-level syndrome or not.

\begin{figure}[htbp]
	\centering
	\begin{subfigure}{0.49\textwidth}
		\includegraphics[width=\textwidth]{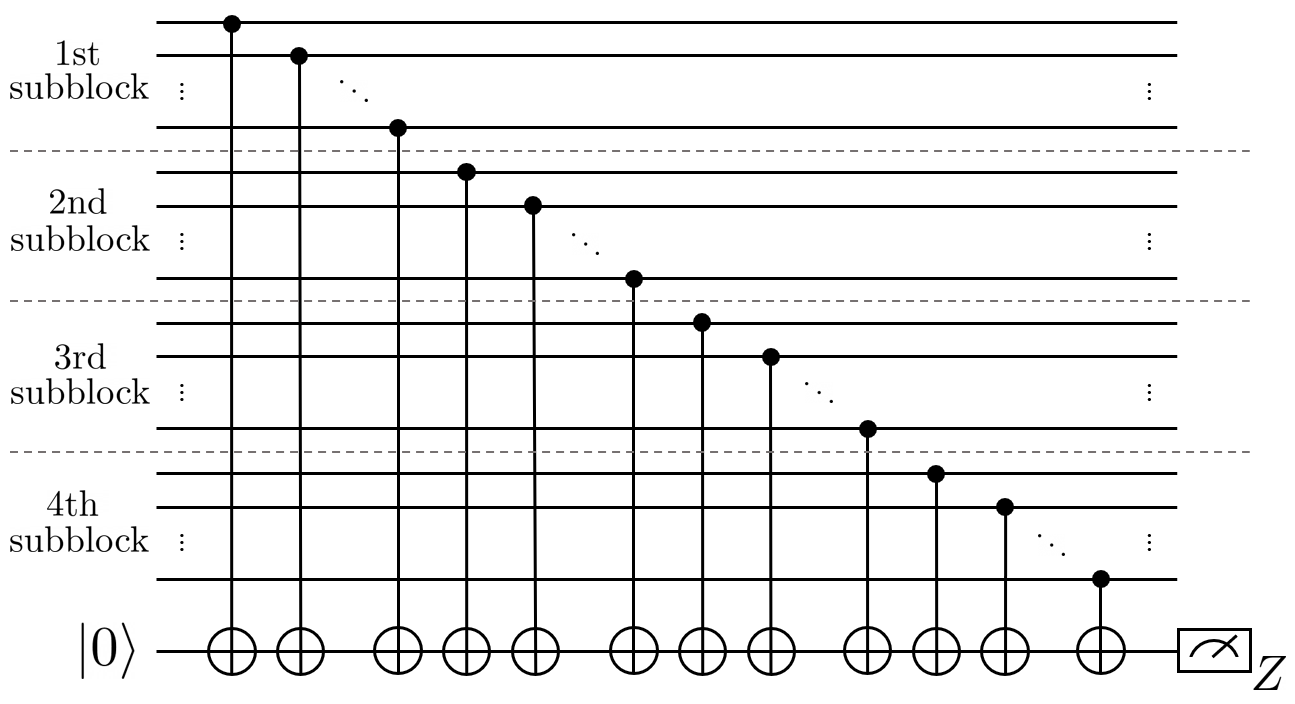}
		\captionsetup{justification=centering}
		\caption{}
		\label{fig:circuit_lv2_normal}
	\end{subfigure}	
	\begin{subfigure}{0.25\textwidth}
		\includegraphics[width=\textwidth]{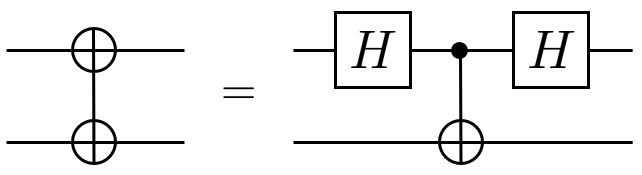}
		\captionsetup{justification=centering}
		\caption{}
		\label{fig:XMeas}
	\end{subfigure}
	\caption{(a) An example of circuit for measuring generator $\tilde{g}^z_1 = \mathbf{ZIZZZII}$. Here we display only the subblocks in which the operator acts nontrivially (the 1st, 2nd, 3rd, and 4th subblocks in the figure correspond to the 1st, 3rd, 4th, and 5th subblocks of $\tilde{g}^z_1$). A circuit for measuring $X$-type operator such as $\tilde{g}^x_1 = \mathbf{XIXXXII}$ can be obtained by replacing each CNOT gate in (a) with the gate illustrated in (b).}
	\label{fig:circuit_normal}
\end{figure}

In particular, let \emph{block parity} $p_x \in \mathbb{Z}_2^7$ be a bitstring, where each bit is the weight parity of the $Z$ error in one subblock, and $0$ and $1$ represent even and odd weights, respectively. Also, define the \emph{triviality} of a subblock to be $0$ or $1$ if the subblock has trivial or nontrivial 1st-level syndrome, and let \emph{block triviality} $\tau_x \in \mathbb{Z}_2^7$ be a $7$-bit string in which the $i$-th bit represents the triviality of the $i$-th subblock.
If the block parity can be accurately determined using the 2nd-level syndrome together with the block triviality (we will elaborate how this can be done later), then we can blockwisely perform WPEC as described in \cref{Def:WPEC} by using the 1st-level syndrome and the weight parity of each subblock.



In this work, we develop an FTEC protocol that uses WPEC to correct high-weight errors arising from up to 3 faults. As an example, consider the measurement of $\tilde{g}^z_1$ using the circuit depicted in \cref{fig:circuit_lv2_normal}. Here we assume that a fault from any two-qubit gate can cause any two-qubit Pauli errors on the qubits where the gate acts nontrivially, and $X$-type and $Z$-type errors can be detected separately. Thus, we may assume that high-weight errors arising from a single CNOT fault is of the form $\mathbf{PIZZZII}$, $\mathbf{IIPZZII}$, $\mathbf{IIIPZII}$, or $\mathbf{IIIIPII}$, where $\mathbf{Z} = Z^{\otimes 7}$, $\mathbf{P}=I^{\otimes 7-m}\otimes Z^{\otimes m}$, and $m\in\{1,\dots,7\}$ (see the analysis of possible errors in \cite{TCL20} for more details). It is not hard to find 2nd-level syndrome, block triviality, and block parity corresponding to each possible error. For example, error $\mathbf{PIZZZII}$ with $m=6$ anticommutes with $g^x_1$ and $\tilde{g}^x_1$, and commutes with the other generators. Thus, its corresponding 2nd-level syndrome, block triviality, and block parity are $(1,0,0), (1,0,0,0,0,0,0)$, and $(0,0,1,1,1,0,0)$, respectively. \cref{tab:ex_1fault} displays all possible high-weight errors arising from a single fault during $\tilde{g}^z_1$ measurement and their corresponding 2nd-level syndrome, block triviality, and block parity. Note that except for the first and the last row (with errors differing by multiplication of a stabilizer), each row has a unique combination of 2nd-level syndrome and block triviality, so the block parity can be determined from the table. Since the 2nd-level syndrome and the block triviality can in turn be obtained from the generator measurements, all possible errors arising from a single CNOT fault during the measurement of $\tilde{g}^z_1$ can be corrected using WPEC. In addition, observe that $\mathbf{ZIZZZII}$ and $\mathbf{I}^{\otimes 7}$ are equivalent up to a multiplication of $\tilde{g}^z_1$ but their block parities are different. Here we can see that multiplying an error with some 2nd-level generators may change its block parity, but its 2nd-level syndrome and block triviality (which is deduced from its 1st-level syndrome) remain the same. In this case, WPEC is still applicable. We say that block parities are equivalent whenever they can be transformed to one another by multiplying the corresponding errors with some stabilizer.


\begin{table}[tbp]
	\begin{center}
		\begin{tabular}{| c | c | c | c | c |}
			\hline
			Form of & \multirow{2}{*}{$m$} & 2nd-level & Block & \multirow{2}{*}{Block parity} \\
			error & & syndrome & triviality & \\
			\hhline{|=|=|=|=|=|}
			\multirow{3}{*}{$\mathbf{PIZZZII}$} & 7 & (0,0,0) & (0,0,0,0,0,0,0) & (1,0,1,1,1,0,0) \\
			\cline{2-5}
			& 2,4,6 & (1,0,0) & (1,0,0,0,0,0,0) & (0,0,1,1,1,0,0) \\
			\cline{2-5}
			& 1,3,5 & (0,0,0) & (1,0,0,0,0,0,0) & (1,0,1,1,1,0,0) \\
			\hline
			\multirow{3}{*}{$\mathbf{IIPZZII}$} & 7 & (1,0,0) & (0,0,0,0,0,0,0) & (0,0,1,1,1,0,0) \\
			\cline{2-5}
			& 2,4,6 & (0,0,1) & (0,0,1,0,0,0,0) & (0,0,0,1,1,0,0) \\
			\cline{2-5}
			& 1,3,5 & (1,0,0) & (0,0,1,0,0,0,0) & (0,0,1,1,1,0,0) \\
			\hline
			\multirow{3}{*}{$\mathbf{IIIPZII}$} & 7 & (0,0,1) & (0,0,0,0,0,0,0) & (0,0,0,1,1,0,0) \\
			\cline{2-5}
			& 2,4,6 & (1,1,1) & (0,0,0,1,0,0,0) & (0,0,0,0,1,0,0) \\
			\cline{2-5}
			& 1,3,5 & (0,0,1) & (0,0,0,1,0,0,0) & (0,0,0,1,1,0,0) \\
			\hline
			\multirow{3}{*}{$\mathbf{IIIIPII}$} & 7 & (1,1,1) & (0,0,0,0,0,0,0) & (0,0,0,0,1,0,0) \\
			\cline{2-5}
			& 2,4,6 & (0,0,0) & (0,0,0,0,1,0,0) & (0,0,0,0,0,0,0) \\
			\cline{2-5}
			& 1,3,5 & (1,1,1) & (0,0,0,0,1,0,0) & (0,0,0,0,1,0,0) \\
			\hline
			$\mathbf{I}^{\otimes 7}$ & - & (0,0,0) & (0,0,0,0,0,0,0) & (0,0,0,0,0,0,0) \\
			\hline
		\end{tabular}
		\caption{All possible forms of data errors arising from a single fault occurred during syndrome measurement using a circuit in \cref{fig:circuit_lv2_normal} (where $\mathbf{P}=I^{\otimes 7-m}\otimes Z^{\otimes m}$). The block parity corresponding to each form of errors can be determined by the 2nd-level syndrome and the block triviality obtained from a full syndrome measurement. By knowing the block parity, high-weight errors can be corrected using WPEC.}
		\label{tab:ex_1fault}%
	\end{center}
\end{table}%
In an actual fault-tolerant protocol, we want to distinguish all possible high-weight errors arising from various types of faults up to 3 faults, including any gate faults, faults during the preparation and measurement of ancilla qubits, and faults during wait time. The circuit construction in \cref{fig:circuit_lv2_normal}, however, might not cause errors that can be distinguished. Note that possible errors arising from CNOT faults heavily depend on the ordering of CNOT gates being used in the measurement circuit. In \cref{sec:Protocol}, we will discuss conditions in which WPEC can be applied. We will also provide a family of circuits with specific CNOT ordering and an FTEC protocol for the \codepar{49,1,9} code which can correct high-weight errors arising from up to 3 faults.

\section{Fault-tolerant error correction protocol for the \codepar{49,1,9} code}
\label{sec:Protocol}%

Fault-tolerant error correction is one of the most essential gadgets for constructing large-scale quantum computers. Every FTEC protocol must satisfy the following conditions:

\begin{definition}{Fault-tolerant error correction \cite{AGP06}}
	
	For $t = \lfloor (d-1)/2\rfloor$, an error correction protocol using a distance-$d$ stabilizer code is \emph{$t$-fault tolerant} if the following two conditions are satisfied:
	\begin{enumerate}
		\item For any input codeword with error of weight $v_{1}$, if $v_{2}$ faults occur during the protocol with $v_{1} + v_{2} \le t$, ideally decoding the output state gives the same codeword as ideally decoding the input state.
		\item If $v$ faults happen during the protocol with $v \le t$, no matter how many errors are present in the input state, the output state differs from any valid codeword by an error of at most weight $v$.
	\end{enumerate}
	\label{Def:FaultTolerantDef}%
\end{definition}
In this work, we develop an FTEC protocol for the \codepar{49,1,9} code that can correct up to 3 faults. The circuits for measuring 1st-level and 2nd-level generators are shown in \cref{fig:circuit_WPEC}. The types of faults being considered include faults that happen to the physical gates, faults during the preparation and measurement of ancilla qubits in the circuits, and faults during wait time.



	

\begin{figure}[htbp]
	\centering
	\begin{subfigure}{0.49\textwidth}
		\includegraphics[width=\textwidth]{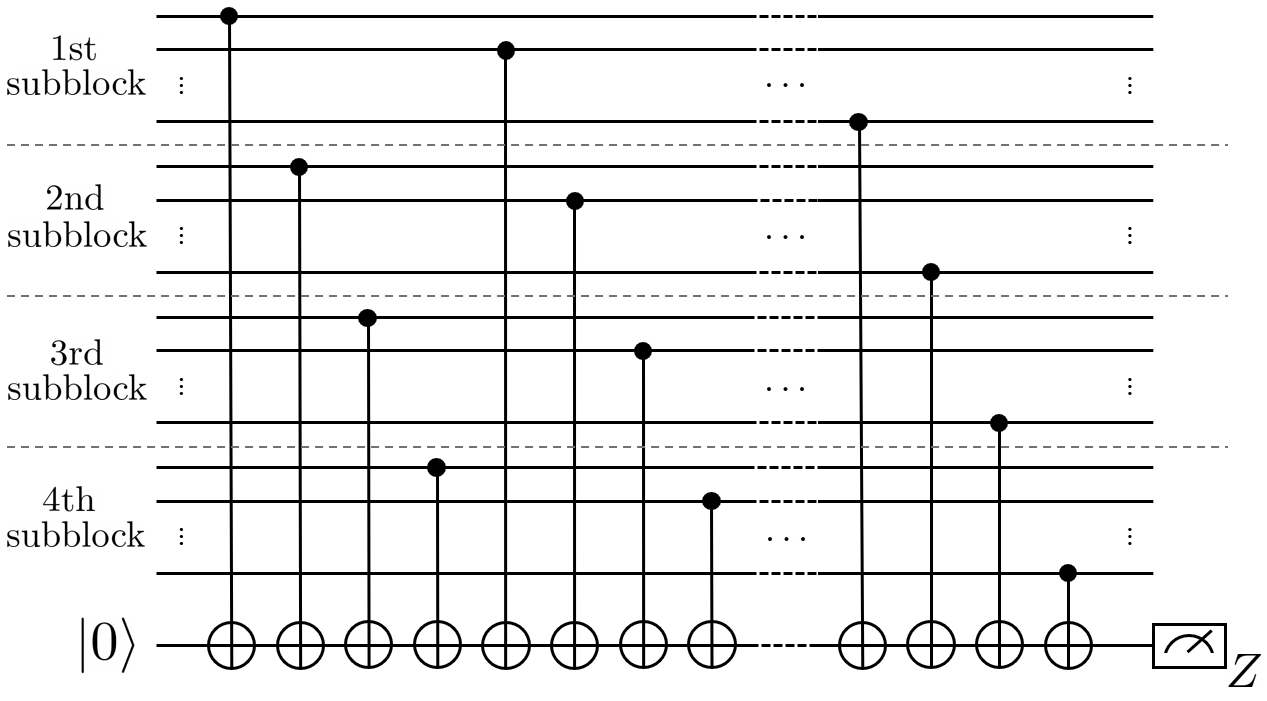}
		\captionsetup{justification=centering}
		\caption{}
		\label{fig:circuit_lv2_permuted}
	\end{subfigure}	
	\begin{subfigure}{0.32\textwidth}
		\includegraphics[width=\textwidth]{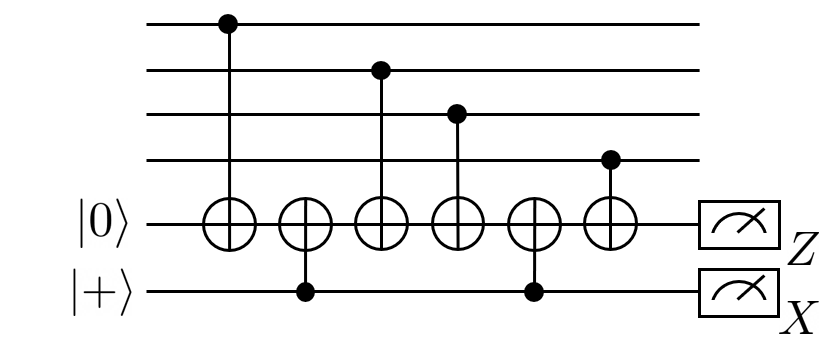}
		\captionsetup{justification=centering}
		\caption{}
		\label{fig:circuit_lv1_flag}
	\end{subfigure}
	\caption{Circuits for measuring 2nd-level and 1st-level generators being used in this work are shown in (a) and (b), respectively. With this gate permutation, the block parity corresponding to every possible high-weight error arising from up to 3 faults can be accurately determined. As such, our protocol can correct up to 3 faults.}
	\label{fig:circuit_WPEC}
\end{figure}

Let \emph{fault combination} be a collection of faults up to 3 faults (which may be of different types and can cause errors of weight much higher than 3 on the data qubits). Our goal is to distinguish all fault combinations that can be confusing and may cause WPEC to fail. Similar to an example of WPEC in \cref{sec:WPEC}, we can categorize all possible fault combinations into subsets by their 2nd-level syndrome and block triviality. The following sufficient condition can determine when the WPEC technique can be applied: 

\begin{claim}{Sufficient condition for WPEC}
	
	Let $\mathcal{F}$ be the set of all possible fault combinations during an FTEC protocol for the \codepar{49,1,9} code and let $\mathcal{F}_k \subseteq \mathcal{F}$ be a subset of fault combinations with the same 2nd-level syndrome and the same block triviality (where $\bigcup_k \mathcal{F}_k = \mathcal{F}$). WPEC is applicable in the FTEC protocol if each $\mathcal{F}_k$ satisfies one of the following conditions: 
	\begin{enumerate}
		\item Data errors from all fault combinations in $\mathcal{F}_k$ give equivalent block parities.
		\pagebreak
		\item Not every data error from a fault combination in $\mathcal{F}_k$ gives the same block parity (or its equivalence), but for each pair of fault combinations in $\mathcal{F}_k$ whose block parities of their data errors are not equivalent, their 1st-level syndromes or flag measurement results (or both) are different.
	\end{enumerate}
	\label{claim:suff_cond}
\end{claim}

\begin{proof}
	Whenever subset $\mathcal{F}_k$ satisfies the first condition in \cref{claim:suff_cond}, we can find a block parity that works for all fault combinations in $\mathcal{F}_k$ using only the 2nd-level syndrome and the block triviality. A correction operator for each fault combination can be found following the definition of WPEC (\cref{Def:WPEC}) using the 1st-level syndrome and the block parity. On the other hand, if $\mathcal{F}_k$ satisfies the second condition in \cref{claim:suff_cond}, a block parity cannot be accurately determined using only the 2nd-level syndrome and the block triviality. Fortunately, with the assistance of the 1st-level syndrome and the flag measurement result, fault combinations that correspond to non-equivalent block parities can be distinguished and the block parity of each fault combination can be found. Similarly, a correction operator for each fault combination can be determined following \cref{Def:WPEC}. 
\end{proof}



Whether possible fault combinations satisfy \cref{claim:suff_cond} or not depends heavily on the ordering of the CNOT gates and the use of flag qubits in the circuits for syndrome measurements. In our FTEC protocol for the \codepar{49,1,9} code, the CNOT gates being used in the circuits for measuring 2nd-level generator are applied in the following ordering:
\begin{equation}
(1,8,15,22,2,9,16,23,\dots,7,14,21,28), \label{eq:CNOTperm}
\end{equation}%
where the numbers 1 to 28 represent the qubits in which $\tilde{g}^z_i$ acts nontrivially. That is, CNOT gates are applied on the first qubit in each subblock for all subblocks, then on the second qubit in each subblock for all subblocks, and so on. The circuit for measuring $\tilde{g}^z_1$ is shown in \cref{fig:circuit_lv2_permuted}. In addition, CNOT gates being used in the circuits for measuring 1st-level generator are in the normal ordering as shown in \cref{fig:circuit_lv1_flag}. Note that there is no flag qubit involved in the measurement of a 2nd-level generator, and there is one flag qubit in the circuit for measuring a 1st-level generator.

Consider the case that there are some faults during $Z$-type generator measurements. Faulty circuits can produce nontrivial flag measurement results and cause error of any weight on the data qubits. Our goal is to detect and correct such an error using the flag measurement results from the faulty circuits, together with 1st-level and 2nd-level syndromes obtained from subsequent syndrome measurements. In particular, let the \emph{flag vector} $\in \mathbb{Z}^{21}_2$ be a bitstring wherein each bit is the flag measurement result from each circuit for measuring $g^z_i$ on each of the 7 subblocks. We define the \emph{cumulative flag vector} $f_x \in \mathbb{Z}^{21}_2$ to be the entry-wise sum of flag vectors (modulo 2) obtained from $g^z_i$ measurements accumulated from the first round up until the current round of measurements (see the protocol described below for the definition of a round of measurements). Cumulative flag vector $f_x$ together with 1st-level syndrome $s_x \in \mathbb{Z}^{21}_2$, 2nd-level syndrome $\tilde{s}_x \in \mathbb{Z}^{3}_2$, and block triviality $\tau_x \in \mathbb{Z}^{7}_2$ from the latest round of measurements will be used for distinguishing all possible fault combinations that can occur during the syndrome measurements as described in \cref{claim:suff_cond}. Using the computer simulation described in \cref{app:sim1}, we can verify that \cref{claim:suff_cond} is satisfied when the number of input errors $v_1$ and the number of faults $v_2$ satisfy $v_1+v_2 \leq 3$. A table of possible data errors and their corresponding $s_x, \tilde{s}_x, \tau_x, f_x$, and block parity $p_x$ (similar to \cref{tab:ex_1fault}) can also be obtained from the simulation. Moreover, the subsets $\mathcal{F}_k$ can be deduced from this table (see \cref{app:sim1} for more details).


Let the \emph{outcome bundle} be the collection of 1st-level syndrome $s = (s_x|s_z)$, 2nd-level syndrome $\tilde{s} = (\tilde{s}_x|\tilde{s}_z)$, block triviality $\tau = (\tau_x|\tau_z)$, and cumulative flag vector $f = (f_x|f_z)$ obtained during a single round of full syndrome measurement, where subscripts $x$ and $z$ denote results corresponding to $X$-type and $Z$-type generator measurements. Using the simulation result together with the fact that $X$-type and $Z$-type errors can be corrected separately, an FTEC protocol for the \codepar{49,1,9} code can be constructed as follows.

\pagebreak


\textbf{FTEC protocol for the \codepar{49,1,9} code}

A full syndrome measurement, or a round of measurements, measure the generators in the following order: measure all $\tilde{g}^z_i$'s, then all $\tilde{g}^x_i$'s, then all $g^z_i$'s, then all $g^x_i$'s. Perform full syndrome measurements until the outcome bundles are repeated 4 times in a row. Afterwards, perform the following error correction procedure:

\begin{enumerate}
	\item Find the subset $\mathcal{F}_k$ corresponding to $\tilde{s}_x$ and $\tau_x$ from the table of possible errors (obtained from the simulation in \cref{app:sim1}).
	\begin{enumerate}
		\item If $\mathcal{F}_k$ satisfies Condition 1 in \cref{claim:suff_cond}, 
		use a block parity of any fault combination in $\mathcal{F}_k$.
		\item If $\mathcal{F}_k$ satisfies Condition 2 in \cref{claim:suff_cond}, use a block parity of any combination in $\mathcal{F}_k$ that corresponds to $s_x$ and $f_x$.
		\item If there is no $\mathcal{F}_k$ from the table of possible errors which corresponds to  $\tilde{s}_x$ and $\tau_x$, use the block parity with all 1's.
	\end{enumerate}	
      \item Let $s_{x,i}$ be the 1st-level syndrome and $w_i$ be the weight parity of the $i$-th subblock. Apply $Z$-type error correction on each subblock as given by \cref{Def:WPEC}.  In particular: 
		\begin{enumerate}
			\item If $s_{x,i}$ is trivial, apply $ZZIZIII$ (logically equivalent to $Z^{\otimes 7}$) to the $i$-th subblock when $w_i$ is odd, or do nothing when $w_i$ is even.
			\item If $s_{x,i}$ is nontrivial, apply $E^z_{wt\mhyphen 1}(s_{x,i})$ to the $i$-th subblock when $w_i$ is odd, or apply $E^z_{wt\mhyphen 2}(s_{x,i})$ when $w_i$ is even.
		\end{enumerate}		
	\item If there is no $\mathcal{F}_k$ from the table of possible errors which corresponds to  $\tilde{s}_x$ and $\tau_x$, further apply the following error correction procedure: find a Pauli operator from $\{\mathbf{ZIIIIII},\mathbf{IZIIIII},\dots,\mathbf{IIIIIIZ}\}$ which corresponds to $\tilde{s}_x$, then apply such an operator (or its logically equivalent operator) to the data qubits.
	
	\item Repeat steps 1--3 but use $\tilde{s}_z$, $s_z$, $\tau_z$, and $f_z$ to deduce the $X$-type error correction ($E^x_{wt\mhyphen 1}(s_{z,i})$, $E^x_{wt\mhyphen 2}(s_{z,i})$, or $XXIXIII$) for each subblock.
\end{enumerate}


Here we will assume that there are at most 3 faults during the protocol and the error is of $Z$ type. The assumption on the number of faults guarantees that the outcome bundles must be repeated 4 times in a row within 16 rounds (the outcome bundle cannot keep changing forever since the number of faults is limited). To verify that the protocol above is 3-fault tolerant, i.e., it satisfies the FTEC conditions in \cref{Def:FaultTolerantDef} with $t=3$ (the \codepar{49,1,9} code acts as a distance-7 code), first let us consider the case that there are no faults during the last round of full syndrome measurement. In this case, the outcome bundle corresponds to the actual data error. From the simulation in \cref{app:sim1}, we know that whenever $v_1+v_2 \leq 3$, one of the conditions in \cref{claim:suff_cond} is satisfied and the block parity can be accurately determined. The operation in Step 2 will give the correct output state, thus both of the FTEC conditions are satisfied. On the other hand, if $v_1+v_2 > 3$ but $v_2 \leq 3$, $\tilde{s}_x$ and $\tau_x$ may not correspond to any error in the table of possible errors. By using the block parity with all 1's, the operation in Step 2 will project the state in each subblock back to the subspace of the \codepar{7,1,3} code, where each subblock has an error equivalent to either $\mathbf{I}$ or $\mathbf{Z}$ after the operation. Afterwards, the operation in Step 3 will project the output state back to the subspace of the \codepar{49,1,9} code. Thus, the second condition in \cref{Def:FaultTolerantDef} is satisfied.


Now, let us consider the case that there are some faults during the last round of full syndrome measurement. The outcome bundle we obtained from the last round may not correspond to the data error since some errors arising during the last round may be undetectable. Since we perform full syndrome measurements until the outcome bundles are repeated 4 times in a row and there are at most 3 faults during the whole protocol, at least one of the last 4 rounds of full syndrome measurement must be correct. From the simulation result in \cref{app:sim1}, the outcome bundle obtained from the last round (which is equal to that obtained from any correct round in the last 4 rounds) can definitely correct the error occurred before the last correct round. Here we want to verify that whenever the last 4 rounds have $v$ faults (where $v \leq 3$), after the last round, the weight of the data error is increased by no more than $v$. This can be verified using the computer simulation described in \cref{app:sim2}. By applying operation in Step 2 (and possibly Step 3) as previously discussed, the output state differs from a valid codeword by an error of weight at most $v$, regardless of the number of input errors. Thus, the second condition in \cref{Def:FaultTolerantDef} is satisfied. Furthermore, whenever $v_1+v_2 \leq 3$, we will obtain an output state which differs from a correct output state by an error of weight at most 3. Therefore, the first condition in \cref{Def:FaultTolerantDef} is also satisfied.


The analysis for $X$-type errors is similar to that of $Z$-type errors. Note that during the measurement of $Z$-type generators, a single gate fault can cause an $X$-type error of weight 1 on the data qubits. This error can be considered as an input error for the measurement of $X$-type generators, thus the same analysis is applicable.




\section{Weight parity error correction for other codes}
\label{sec:WPEC_Golay}%

Besides the Steane code, we find that the WPEC technique is applicable to the \codepar{23,1,7} Golay code \cite{Steane03}, which is a perfect CSS code of distance 7. The \codepar{23,1,7} Golay code can correct up to 3 errors and can be constructed from the parity check matrix of the classical $[23,12,7]$ Golay code \cite{MS77}. In cyclic form, the \codepar{23,1,7} Golay code can be constructed from the parity check matrix,
\begingroup
\setlength\arraycolsep{2pt}	
\begin{equation}
\begin{pmatrix}
1 & 1 & 1 & 1 & 1 & 0 & 0 & 1 & 0 & 0 & 1 & 0 & 1 & 0 & 0 & 0 & 0 & 0 & 0 & 0 & 0 & 0 & 0 \\
0 & 1 & 1 & 1 & 1 & 1 & 0 & 0 & 1 & 0 & 0 & 1 & 0 & 1 & 0 & 0 & 0 & 0 & 0 & 0 & 0 & 0 & 0 \\
0 & 0 & 1 & 1 & 1 & 1 & 1 & 0 & 0 & 1 & 0 & 0 & 1 & 0 & 1 & 0 & 0 & 0 & 0 & 0 & 0 & 0 & 0 \\
0 & 0 & 0 & 1 & 1 & 1 & 1 & 1 & 0 & 0 & 1 & 0 & 0 & 1 & 0 & 1 & 0 & 0 & 0 & 0 & 0 & 0 & 0 \\
0 & 0 & 0 & 0 & 1 & 1 & 1 & 1 & 1 & 0 & 0 & 1 & 0 & 0 & 1 & 0 & 1 & 0 & 0 & 0 & 0 & 0 & 0 \\
0 & 0 & 0 & 0 & 0 & 1 & 1 & 1 & 1 & 1 & 0 & 0 & 1 & 0 & 0 & 1 & 0 & 1 & 0 & 0 & 0 & 0 & 0 \\
0 & 0 & 0 & 0 & 0 & 0 & 1 & 1 & 1 & 1 & 1 & 0 & 0 & 1 & 0 & 0 & 1 & 0 & 1 & 0 & 0 & 0 & 0 \\
0 & 0 & 0 & 0 & 0 & 0 & 0 & 1 & 1 & 1 & 1 & 1 & 0 & 0 & 1 & 0 & 0 & 1 & 0 & 1 & 0 & 0 & 0 \\
0 & 0 & 0 & 0 & 0 & 0 & 0 & 0 & 1 & 1 & 1 & 1 & 1 & 0 & 0 & 1 & 0 & 0 & 1 & 0 & 1 & 0 & 0 \\
0 & 0 & 0 & 0 & 0 & 0 & 0 & 0 & 0 & 1 & 1 & 1 & 1 & 1 & 0 & 0 & 1 & 0 & 0 & 1 & 0 & 1 & 0 \\
0 & 0 & 0 & 0 & 0 & 0 & 0 & 0 & 0 & 0 & 1 & 1 & 1 & 1 & 1 & 0 & 0 & 1 & 0 & 0 & 1 & 0 & 1
\end{pmatrix},
\nonumber
\end{equation}%
\endgroup 
which can be generated from the check polynomial $h(x) = x^{12}+x^{10}+x^7+x^4+x^3+x^2+x+1$ \cite{MS77}. The $i$-th $Z$-type (or $X$-type) generator of this code will be denoted as $g^z_i$ (or $g^x_i$) where $i=1,\dots,11$. The logical $X$ and logical $Z$ operators of this code can be chosen to be $\bar{X}=X^{\otimes 23} S$ and $\bar{Z}=Z^{\otimes 23} T$ for any stabilizer operators $S, T$.

Similar to the \codepar{7,1,3} code, we can prove the equivalence of errors with the same syndrome and the same weight parity as follows:



\begin{claim}{Logical equivalence of errors with the same syndrome for the \codepar{23,1,7} Golay code}
	
	Suppose $E_1,E_2$ are arbitrary $Z$-type errors (of any weights) on the \codepar{23,1,7} code with the same syndrome. Then, $E_1$ and $E_2$ have the same weight parity iff $E_1 = E_2S$ for some stabilizer $S$.
	\label{claim:equiv_Golay}%
\end{claim}

\begin{proof}
	We can verify that every $Z$-type stabilizer in the stabilizer group of the \codepar{23,1,7} code has even weight, and every logical $Z$ operator has odd weight. The rest of the proof follows the proof of \cref{claim:equiv_Steane}.
\end{proof}

Let us consider $Z$-type error correction for the \codepar{23,1,7} code. Since the code is a perfect CSS code of distance 7, for each $s_x \in \mathbb{Z}_2^{11}$, $(s_x|0...0)$ is the syndrome of a unique $Z$-type error of weight $\leq 3$. Suppose that a codeword is corrupted by a $Z$-type error with syndrome $s_x$. If we apply the minimal weight error correction corresponding to $s_x$, we sometimes obtain the codeword with undesirable logical $Z$ operator. Fortunately, by knowing the weight parity of the error, the WPEC technique can be applied. The error correction procedure for the \codepar{23,1,7} code is defined as follows:


\begin{definition}{Weight parity error correction for the \codepar{23,1,7} Golay code}
	
	Suppose a $Z$-type error $E$ occurs to a codeword of the \codepar{23,1,7} code. Let $s_x$ and $w$ be the syndrome and the weight of $E$, and let $E^z_{min}(s_x)$ be the unique minimal weight error correction corresponding to the syndrome $s_x$. The following procedure is called \emph{weight parity error correction (WPEC)}:
	\begin{enumerate}
		\item If $E^z_{min}(s_x)$ has even weight (0 or 2), apply $E^z_{min}(s_{x})$ to the data qubits whenever $w$ is even, or apply any $Z$-type operator $P$ that has odd weight and corresponds to $s_x$ to the data qubits whenever $w$ is odd.
		\item If $E^z_{min}(s_x)$ has odd weight (1 or 3), apply $E^z_{min}(s_{x})$ to the data qubits whenever $w$ is odd, or apply any $Z$-type operator $P$ that has even weight and corresponds to $s_x$ to the data qubits whenever $w$ is even.
	\end{enumerate}
	\label{Def:WPEC_Golay}%
\end{definition}

Note that the \codepar{23,1,7} Golay code can be made cyclic, thus it can distinguish high-weight errors in consecutive form \cite{TCL20}. \cref{claim:equiv_Golay} together with the cyclic property give us some possibilities to construct an FTEC protocol for the \codepar{529,1,49} concatenated Golay code in the same way as what we have done for the \codepar{49,1,9} code. We expect that our technique can lead to a protocol which can correct a large number of faults and will compare well with other FTEC schemes. To reach this goal, syndrome extraction circuits with appropriate permutation of gates (and possibly with flag qubits) must be found so that conditions similar to those in \cref{claim:suff_cond} are satisfied. The search for such circuits with careful numerical verification of fault tolerance is a challenging and interesting future research direction.

The WPEC technique may also apply to the code obtained from concatenating the Steane code to the $k$-th level, e.g., the \codepar{7^k,1,3^k} code.
Since the $k$th-level Steane code is a distance-3 code, we expect that a block of errors in the $(k-1)$-th level can be determined and corrected using the syndrome and the block parity defined at the $k$-th level. Again, however, appropriate syndrome extraction circuits must be found, which is beyond the scope of this work.

\section{Discussions and conclusions}
\label{sec:Discussion}%

In this work, we prove the logical equivalence between errors of any weight on 7 qubits which have the same weight parity and correspond to the same error syndrome when error detection is performed by the \codepar{7,1,3} code in \cref{claim:equiv_Steane}. From this result, we introduce the WPEC technique in \cref{Def:WPEC}, which can correct errors of any weight on 7 qubits whenever their weight parity is known. We show that the WPEC technique can be extended to error correction in subblocks of the \codepar{49,1,9} code, and we prove the sufficient condition for WPEC in \cref{claim:suff_cond}. Afterwards, we provide a family of circuits and an FTEC protocol for the \codepar{49,1,9} code which can correct up to 3 faults. We also point out that the WPEC technique seems applicable to FTEC schemes for other codes such as the concatenated Golay code and concatenated Steane code with more than 2 levels of concatenation.

Since the FTEC protocol provided in this work satisfies the definition of FTEC in \cref{Def:FaultTolerantDef} with $t=3$, we can guarantee that the logical error rate is suppressed to $O(p^4)$ whenever the physical error rate is $p$ under the random Pauli noise model. Note that we did not use the full ability of a code with distance 9 which, in principle, can correct up to 4 errors. In terms of error suppression, our FTEC protocol is as good as typical FTEC protocols for a concatenated code which are constructed by replacing each physical qubit with a code block and replacing each physical gate with the corresponding logical gate \cite{AGP06}.

One major advantage of our FTEC protocol is that only 2 ancillas are needed: one ancilla for a syndrome measurement result and another ancilla for a flag measurement result (assuming that the qubit preparation and measurement are fast compared to the gate operation time). As a result, our protocol requires 51 qubits in total. The number of required qubits is very low compared to other FTEC protocols for the \codepar{49,1,9} code; the FTEC schemes in \cite{YK17,CR17a} extended to the \codepar{49,1,9} code require 63 qubits in total (the minimum number of required ancillas is 14 assuming that they are recyclable). Meanwhile, the FTEC protocol in \cite{Reichardt18} which extracts multiple syndromes at once encodes 2 logical qubits and requires no ancilla, but needs to work on two code blocks of 98 qubits in total. Our protocol might not have the fewest total number of qubits compared with other protocols for a different code which can correct up to 3 faults; for example, the flag FTEC protocol in \cite{CR20} applying to the \codepar{37,1,7} 2D color code requires 45 qubits in total. Nevertheless, our work provides a substantial improvement over other FTEC protocols for a \emph{concatenated} code, an approach that is still advantageous in some circumstances. Furthermore, the use of weight parities in error correction may be extended to other families of codes as well \cite{TL21}. We believe that if the protocol requires fewer ancillas, the number of total locations will decrease, which can result in higher accuracy threshold. However, a simulation with careful analysis is required for the accuracy threshold calculation, thus we leave this for future work.

The protocol in \cref{sec:Protocol} which can correct up to 3 faults exploits two techniques: the flag technique which partitions set of possible errors using flag measurement results, and the WPEC technique which corrects errors of any weight using their syndromes and weight parities. It should be emphasized that flag ancillas are not necessarily required for a protocol exploiting WPEC technique; we find that a protocol which uses circuits similar to a circuit in \cref{fig:circuit_lv2_permuted} for 2nd-level syndrome measurements and uses non-flag circuits for 1st-level measurements can correct up to 2 faults.


We point out that the permutation of CNOT gates in the syndrome extraction circuits that make the protocol satisfies \cref{claim:suff_cond} is not unique. We choose the permutation in \cref{eq:CNOTperm} by using the fact that a CSS code constructed from classical cyclic codes can distinguish high-weight errors in the consecutive form \cite{TCL20}. In particular, the circuit is designed in the way that high-weight errors arising in each subblock can be determined by the underlying \codepar{7,1,3} code in cyclic form. We did not prove the optimality of the choice of gate permutation in our protocol, so an FTEC protocol for the \codepar{49,1,9} code with only one ancilla or a protocol that can correct up to 4 faults might be possible.

Last, we note that the WPEC technique introduced in this work is not limited to the \codepar{49,1,9} code. In \cref{sec:WPEC_Golay}, we prove the logical equivalence of errors with the same syndrome and weight parity for the \codepar{23,1,7} Golay code in \cref{claim:equiv_Golay} and provide a WPEC scheme in \cref{Def:WPEC_Golay}, which shows that WPEC can correct some high-weight errors in a subblock of the \codepar{529,1,49} concatenated Golay code. In addition, we expect that WPEC can be applied to any concatenated Steane code with more than 2 levels of concatenation in a similar fashion. However, circuits and a protocol must be carefully designed so that the full error correction ability of the code can be achieved. Another interesting future direction would be extending the WPEC technique to other families of quantum codes.









\section{Acknowledgements}

We thank Christopher Chamberland, Rui Chao, Narayanan Rengaswamy, and Michael Vasmer for interesting comments and suggestions. T.T. acknowledges the support of The Queen Sirikit Scholarship under The Royal Patronage of Her Majesty Queen Sirikit of Thailand. D.L. is supported by an NSERC Discovery grant. The Perimeter Institute is supported in part by the Government of Canada and the Province of Ontario.

\bibliographystyle{ieeetr}
\bibliography{bibtex_FT}

\appendix

\section{Simulation of possible faults during the FTEC protocol assuming that the last round of full syndrome measurement has no faults}
\label{app:sim1}%

As discussed in \cref{sec:Protocol}, in order to verify that the FTEC protocol for the $\codepar{49,1,9}$ code satisfies the FTEC conditions in \cref{Def:FaultTolerantDef}, we consider two separate cases: the case that there are some faults during the last round of full syndrome measurement, and the case that there are not. In this section, we provide details of a simulation to show that whenever the number of faults is at most 3 and none of the faults occurs during the last round, all possible fault combinations satisfy \cref{claim:suff_cond} and our protocol can correct errors on the data qubits.


In our protocol, we will perform full syndrome measurements until the outcome bundles are repeated 4 times in a row. Since there are at most 3 faults, the repetition condition will be satisfied within 16 rounds of full syndrome measurement. In this simulation, we assume that the last round of measurement has no faults, thus the high-weight error on the data qubits arising from at most 3 faults is accumulated from up to 15 rounds. We will use the outcome bundle (syndromes and flag vector) obtained from the last round to determine the fault combination that cause the error so that the corresponding weight parity can be found and the WPEC can be done.

We first define mathematical objects being used in our simulation. Let \emph{fault} be an object with two associated variables: \emph{Pauli error} defined on the code block of 49 qubits arising from the fault, and \emph{flag vector} $\in \mathbb{Z}^{21}_2$ which indicates the flag measurement results associated with the fault. There are 4 types of possible faults: faults during wait time (denoted by $W$), faults arising from the measurement of 1st-level and 2nd-level generators (denoted by $G_1$ and $G_2$, respectively), and flag measurement faults (denoted by $F$). A \emph{fault combination} can be constructed by combining faults of same or different types up to 3 faults, i.e., multiplying their Pauli errors and adding their flag vectors. The errors on the input codeword can be considered as wait time faults in which associated Pauli errors do not propagate to other data qubits during the FTEC protocol. In addition, the $X$-type errors on the data qubits arising from the faults during the measurement of $Z$-type generators can be considered as wait time faults during the measurement of subsequent $X$-type generators, in which our simulation is also applicable. (Since the last round of measurement has no faults, we can assume that the syndromes obtained from the last round are correct and the syndrome measurement faults can be neglected.)

Next, we define \emph{fault set} as follows: for faults of type $G_1$ (or type $G_2$), we denote $F^{G_1}_{i,j}$ (or $F^{G_2}_{i,j'}$) to be sets of possible $G_1$ (or $G_2$) faults arising from a circuit for measuring $g^z_j$, $j=1,\dots,21$ (or $\tilde{g}^z_{j'}$, $j'=1,2,3$) where the number of faults is $i \in \{0,1,2,3\}$ ($g^z_j$ refers to the generator $g^z_{(j-1)\text{ mod }3 +1}$ on the $\lceil j/3 \rceil$-th subblock). Also, we denote $F^{W}_{i}$ and $F^{F}_{i}$ to be sets of possible faults of type $W$ and $F$, respectively, where the number of faults is $i\in \{0,1,2,3\}$. In addition, we define \emph{fault set combination} to be a set of fault sets up to 3 sets.

Last, let $v_{G_1},v_{G_2},v_{W},v_{F}$ be the number of faults of type $G_1,G_2,W,$ and $F$, respectively. $(v_{G_1},v_{G_2},v_{W},v_{F})$ that satisfies $v_{G_1}+v_{G_2}+v_{W}+v_{F} \leq 3$ is called \emph{fault number combination}. 

With the definitions of fault, fault combination, fault set, fault set combination, and fault number combination, now we are ready to describe the simulation.

\textbf{Pseudocode for a simulation of possible faults assuming that the last round of full syndrome measurement has no faults}
\begin{enumerate}
	\item Construct fault sets $F^{G_1}_{i,j}, F^{G_2}_{i,j'}, F^{W}_{i},$ and $F^{F}_{i}$ for all $i=0,1,2,3$, $j=1,\dots,21$, $j'=1,2,3$.
	\item Construct all possible fault number combinations that satisfy $v_{G_1}+v_{G_2}+v_{W}+v_{F} \leq 3$.
	\item For each $(v_{G_1},v_{G_2},v_{W},v_{F})$, find all possible fault set combinations from $v_{G_1},v_{G_2},v_{W},v_{F}$. Note that if $v_{G_1}$ is 2, the fault set combination can have $F^{G_1}_{i,j}$ and $F^{G_1}_{i',j'}$ with $i=i'=1$, or have $F^{G_1}_{i,j}$ with $i=2$. Also, if $v_{G_1}$ is 3, the fault set combination can have $F^{G_1}_{i,j}$, $F^{G_1}_{i',j'}$, and $F^{G_1}_{i'',j''}$ with $i=i'=i''=1$, or have $F^{G_1}_{i,j}$ and $F^{G_1}_{i',j'}$ with $i=2,i'=1$, or have $F^{G_1}_{i,j}$ with $i=3$. The same goes for $v_{G_2}$. 
	\begin{enumerate}
		\item For each fault set combination, find all possible fault combinations. Each fault combination can be found by picking one fault from each fault set (up to 3 sets) in the fault set combination, then combine the faults to get the Pauli error $E$ and the cumulative flag vector $f_x$ associated with the fault combination.
		\item For each fault combination, find 1st-level syndrome $s_x$, 2nd-level syndrome $\tilde{s}_x$, block triviality $\tau_x$, and block parity $p_x$ from the associated Pauli error $E$. Store $(s_x,\tilde{s}_x,\tau_x,f_x,p_x)$ for each fault combination in a lookup table. \label{step:table}
	\end{enumerate}
	\item After the lookup table is complete, categorize fault combinations by their 2nd-level syndromes and block trivialities in order to get $\mathcal{F}_k$'s as in \cref{claim:suff_cond}.
	\item For each $\mathcal{F}_k$, verify whether Condition 1 or 2 in \cref{claim:suff_cond} is satisfied.
\end{enumerate}

From the simulation above, we find that all possible fault combinations satisfy \cref{claim:suff_cond}. That is, for each fault combination, we can determine the weight parity from the outcome bundles obtained from the last round of full syndrome measurement by looking at the table constructed in Step \ref{step:table}. The weight parity can be later used to perform WPEC on the code block. With this simulation result, we can verify our FTEC protocol for the \codepar{49,1,9} code satisfies FTEC conditions as previously discussed in \cref{sec:Protocol}.

\section{Simulation of possible faults during the FTEC protocol assuming that the last round of full syndrome measurement has some faults}
\label{app:sim2}%

In \cref{app:sim1}, we describe the simulation of possible faults during the FTEC protocol for the \codepar{49,1,9} code which is applicable to the case that there are no faults during the last round of full syndrome measurement. In this section, we will extend the ideas and construct a simulation of possible faults for the case that some faults occur during the last round.

As previously described, we will perform full syndrome measurements in the protocol until the outcome bundles are repeated 4 times in a row. Now, suppose that the last round of full syndrome measurement has some faults. In this case, we cannot be sure whether the outcome bundle from the last round exactly corresponds to the error in the data qubits. Fortunately, since there are at most 3 faults during the whole protocol, at least one outcome bundle obtained from the last 4 rounds must be correct. Note that the outcome bundles from the last 4 rounds are identical. From the simulation result discussed in \cref{app:sim1}, the outcome bundle from the last round can be used to correct the data error occurred before \emph{any} correct round using the WPEC technique (see \cref{fig:sim2_p1} for more details). The goal of the simulation in this section is to verify that all possible fault combinations which can happen after the last correct round give data error of weight no more than 3. 
\begin{figure}[htbp]
	\centering
	\includegraphics[width=0.48\textwidth]{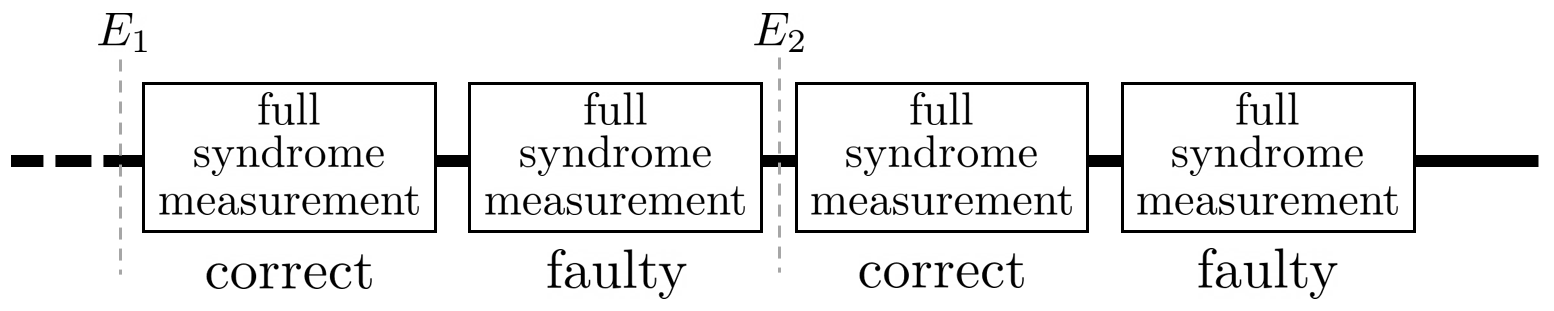}
	\caption{At least one of the last 4 rounds of full syndrome measurement is correct since there are at most 3 faults. Because the outcome bundles from the last 4 rounds are identical, the outcome bundle from the last round can be used in WPEC to correct both errors $E_1$ and $E_2$ (even though $E_1$ and $E_2$ may not be equal).}
	\label{fig:sim2_p1}
\end{figure}

A straightforward way to verify the claim above is to find all possible fault combinations and check the weight of their associated Pauli errors. Unfortunately, this process requires many computational resources. Thus, we will use ``relaxed conditions'' for the verification instead; for each fault combination, if the associated Pauli error and flag vector satisfy all relaxed conditions, the fault combination will be marked (indicating that the fault combination might cause the protocol to fail). We want to make sure that for all fault combination that can cause the protocol to fail (i.e., its associated error has weight more than 3), the fault combination will be marked. Note that some fault combinations may be marked by the relaxed conditions but will not cause the protocol to fail. For this reason, all of the marked fault combinations must be examined after the simulation is done.

We should note that the order of generator measurements is important for the fault tolerance of our FTEC protocol. Consider the protocol description in \cref{sec:Protocol} in which we measure generator measurements in the following order during a single round of full syndrome measurement: measuring all $\tilde{g}^z_i$'s, then all $\tilde{g}^x_i$'s, then all $g^z_i$'s, then all $g^x_i$'s. Let us first consider the errors that can be caught by the $g^x_i$ measurements of the last round. Observe that all $Z$-type data errors that arise before the $g^x_i$ measurements of the last round will be evaluated by the 1st-level syndrome $s_x$. However, some faults during $g^x_i$ measurements of the last round may cause $X$-type or $Z$-type errors that will not be caught by any syndrome. Without loss of generality, we will construct a simulation using an assumption that faults before the $g^x_i$ measurements of the last round can cause only $Z$-type errors, and faults during or after the $g^x_i$ measurements can cause $X$-type or $Z$-type errors. The simulation is also applicable to the case of $g^z_i$ measurements.

Let $E$, $\tilde{E}_a$, and $\tilde{E}_b$ be data errors arising from faults occurred before the last correct round among the last 4 rounds, faults occurred after the last correct round but before the $g^x_i$ measurements of the last round, and faults occurred during or after the $g^x_i$ measurements of the last round. The errors can be illustrated as follows:
\begin{equation}
	\includegraphics[width=0.48\textwidth]{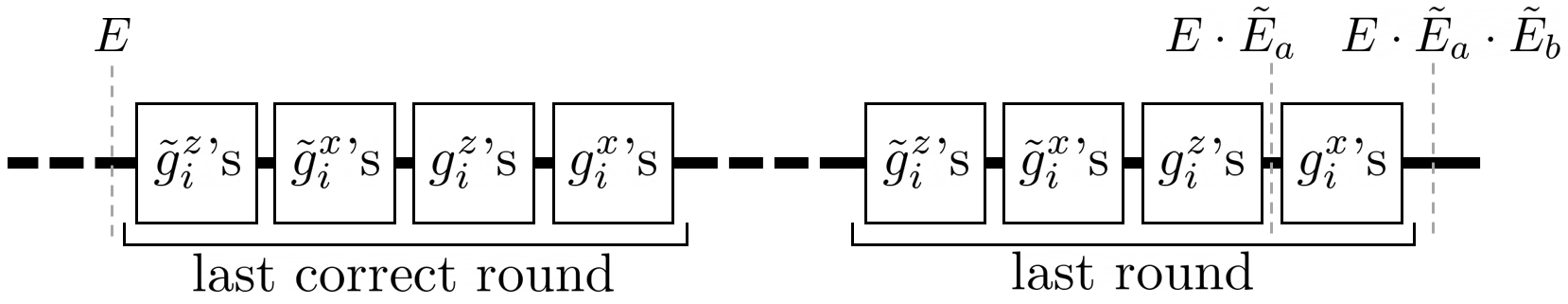} \nonumber
\end{equation}


\noindent The outcome bundle obtained from the last round is equal to the outcome bundle obtained from the correct round and can be used to correct $E$. Thus, we would like to mark every fault combination that can occur after the correct round, corresponds to the trivial outcome bundle (since the outcome bundle obtained from the last round is the same as that obtained from the correct round), and corresponds to a Pauli error of weight more than 3. In particular, our relaxed conditions will examine 3 objects for each fault combination: the 1st-level syndrome, the cumulative flag vector, and the weight of the Pauli error.

The mathematical objects being used in this simulation are similar to those defined in \cref{app:sim1}. In addition, we will consider syndrome measurement faults (denoted by $S$) as another type of faults in this simulation since we will assume that the syndrome measurement during the last 4 rounds can be faulty. Also, let $v_{G_{1a}}$ be the number of $G_1$ faults that occur before the $g^x_i$ measurements of the last round, and let $v_{G_{1b}}$ be the number of $G_1$ faults that occur during or after the $g^x_i$ measurements of the last round. \emph{Fault number combination} is a tuple $(v_{G_{1a}},v_{G_{1b}},v_{G_2},v_W,v_F,v_S)$ that satisfies $v_{G_{1a}}+v_{G_{1b}}+v_{G_2}+v_W+v_F+v_S\leq 3$.

For the first relaxed condition, let us first assume that none of the faults of type $W$ occurs before or during the $g^x_i$ measurements of the last round. For each $(v_{G_{1a}},v_{G_{1b}},v_{G_2},v_W,v_F,v_S)$, error $\tilde{E}_a$ will be constructed from possible fault combinations that correspond to $v_{G_{1a}}$ and $v_{G_2}$. We will mark every fault combination whose associated $\tilde{E}_a$ gives a 1st-level syndrome that has Hamming weight no more than $v_S$ (where the Hamming weight is the number of 1's in a bitstring). This is because each fault of type $S$ can alter at most 1 syndrome bit. Now let us consider the case that some faults of type $W$ occurs before or during the $g^x_i$ measurements. Each $W$ fault (which corresponds to error of weight 1) can change at most 3 bits of $s_x$, but the change will affect only the subblock in which the fault acts nontrivially. We will define function $\sigma(\tilde{E}_a,v_W)$ by the following calculation: 
\begin{enumerate}
	\item Find the 1st-level syndrome of $\tilde{E}_a$ and calculate the Hamming weight of the syndrome for each subblock.
	\item Sort the Hamming weights from all subblocks. The function value is the the sum of the $7-v_W$ smallest Hamming weights.
\end{enumerate}
The value of $\sigma(\tilde{E}_a,v_W)$ is the minimum Hamming weight of the 1st-level syndrome when $v_W$ faults of type $W$ occur. Taking all fault types into account, our first relaxed condition becomes
\begin{equation}
\sigma(\tilde{E}_a,v_W) \leq v_S. \label{eq:relaxed_1}
\end{equation}

For the second relaxed condition, we will consider the cumulative flag vector associated with each fault combination. Note that a flag measurement result will be obtained during any $g^x_i$ or $g^z_i$ measurement. Let $f = (f_x|f_z)$ denote the cumulative flag vector associated with each fault combination, and let $h(f)$ denote the Hamming weight of $f$. Since each fault of type $F$ can alter at most 1 bit of $f$, our second relaxed condition becomes,
\begin{equation}
h(f) \leq v_F. \label{eq:relaxed_2}
\end{equation}

For the third relaxed condition, we will consider the weight of the Pauli error associated with each fault combination. The weight is evaluated at the end of the protocol where the resulting error is caused by all faults of type $G_1$, $G_2$, and $W$ (errors arising during or after the $g^x_i$ measurements of the last round can be $X$-type or $Z$-type). If $W$ faults do not occur before or during the $g^x_i$ measurements at the last round, the weight of the resulting error is the weight of $\tilde{E}_a\cdot\tilde{E}_b$. If they do, each $W$ fault can increase the total weight by at most 1. Hence, our third condition becomes,
\begin{equation}
\mathrm{wt}(\tilde{E}_a\cdot\tilde{E}_b)+v_W > 3. \label{eq:relaxed_3}
\end{equation}
Note that the weight of $\tilde{E}_a\cdot\tilde{E}_b$ can be reduced by multiplication of some stabilizer, and the fault combination will not be marked unless \cref{eq:relaxed_3} is satisfied for all choice of stabilizer.

Using the relaxed conditions in \cref{eq:relaxed_1,eq:relaxed_2,eq:relaxed_3}, our simulation to verify that all possible data errors arising after the correct round have weight no more than 3 can be constructed as follows:

\textbf{Pseudocode for a simulation of possible faults assuming that the last round of full syndrome measurement has some faults}
\begin{enumerate}
	\item Construct fault sets $F^{G_1}_{i,j}, F^{G_2}_{i,j'}$ for all $i=0,1,2,3$, $j=1,\dots,21$, $j'=1,2,3$.
	\item Construct all possible fault number combinations that satisfies $v_{G_{1a}}+v_{G_{1b}}+v_{G_2}+v_W+v_F+v_S\leq 3$.
	\item For each $(v_{G_{1a}},v_{G_{1b}},v_{G_2},v_W,v_F,v_S)$, construct all possible fault set combinations from only $v_{G_{1a}}$, $v_{G_{1b}}$, and $v_{G_2}$. During the construction of each fault set combination, label fault sets that come from $v_{G_{1a}}$ or $v_{G_2}$ with letter $a$, and label fault sets that come from $v_{G_{1b}}$ with letter $b$. Note that if $v_{G_{1a}}$ is 2, the fault set combination can have $F^{G_1}_{i,j}$ and $F^{G_1}_{i',j'}$ with $i=i'=1$, or have $F^{G_1}_{i,j}$ with $i=2$. Also, if $v_{G_{1a}}$ is 3, the fault set combination can have $F^{G_1}_{i,j}$, $F^{G_1}_{i',j'}$, and $F^{G_1}_{i'',j''}$ with $i=i'=i''=1$, or have $F^{G_1}_{i,j}$ and $F^{G_1}_{i',j'}$ with $i=2,i'=1$, or have $F^{G_1}_{i,j}$ with $i=3$. The same goes for $v_{G_{1b}}$ and $v_{G_2}$.
	\begin{enumerate}
		\item For each fault set combination, find all possible fault combinations. Each fault combination can be found by picking one fault from each fault set (up to 3 sets) in the fault set combination. $\tilde{E}_a$ associated with each fault combination can be found by combining only faults from fault sets with label $a$, while $f$ and $\tilde{E}_a \cdot \tilde{E}_b$ can be found by combining faults from all fault sets.	
		\begin{enumerate}
			\item For each fault combination, if \cref{eq:relaxed_1,eq:relaxed_2,eq:relaxed_3} are all satisfied, the fault combination will be marked. Note that for \cref{eq:relaxed_3}, the weight of $\tilde{E}_a \cdot \tilde{E}_b$ must be minimized by stabilizer multiplication.
		\end{enumerate}
	\end{enumerate}	
		
	
\end{enumerate}

From the simulation above, we find that there are 6 fault combinations which are marked by the relaxed conditions in \cref{eq:relaxed_1,eq:relaxed_2,eq:relaxed_3}. All of them correspond to the case that $v_{G_2} = 1$, $v_W = 2$, $v_{G_{1a}},v_{G_{1b}},v_F,v_S=0$, and their associated Pauli errors are trivial on 5 subblocks and have either $IIIIIIZ$ or $ZIIIIII$ on 2 subblocks. We find that $IIIIIIZ$ and $ZIIIIII$ correspond to 1st-level syndrome $(001)$ and $(100)$, respectively. Since $v_S = 0$, the associated 1st-level syndrome must be trivial whenever errors from $W$ faults are taken into account. This can happen only when errors from $W$ faults cancel with the aforementioned Pauli error, which means that the resulting error has weight 0. As a result, we find that all of the marked fault combinations cannot cause data error of weight higher than 3. Similar simulations can be done to show that whenever $v$ faults occur where $v=0,1,2$, the weight of the output error is at most $v$. This result verifies that the FTEC protocol for the \codepar{49,1,9} code satisfies FTEC conditions as previously discussed in \cref{sec:Protocol}.

\end{document}